\PassOptionsToPackage{dvipsnames}{xcolor}
\documentclass[
  preprint,
 showpacs,
 showkeys,
 preprintnumbers,
 amsmath,amssymb,
 aps,
  pra,
  longbibliography,
 floatfix,
 ]{revtex4-2}

\usepackage{mathptmx}

\usepackage{amssymb,amsthm,amsmath}

\usepackage{mathbbol}

\usepackage{amsfonts}
\usepackage{enumerate}
\usepackage{graphicx}

\newcommand{\stb}{\mathop{\mathrm{stab}}\nolimits}
\newcommand{\enc}{\mathop{\mathrm{enc}}}

\newcommand{\fix}{\mathop{\mathrm{fix}}}
\newcommand{\orb}{\mathop{\mathrm{type}}}

\newtheorem{prop}{Lemma}
\newtheorem{prop2}[prop]{Proposition}
\newtheorem{def2}[prop]{Definition}

\usepackage{tikz}
\usetikzlibrary{calc,math}
\usepackage {pgfplots}
\pgfplotsset {compat=1.8}
\usepackage{graphicx}

\usepackage{xcolor}
\usepackage{ulem}

\usepackage{hyperref}
\hypersetup{
    colorlinks,
    linkcolor={blue},
    citecolor={red!75!black},
    urlcolor={blue}
}

\begin{document}

\title{Orbits of One-Dimensional Cellular Automata Induced by Symmetry Transformations}

\author{Martin Schaller}
\email{martin.roman.schaller@gmail.com}

\affiliation{Vienna, Austria}

\author{Karl Svozil}
\email{karl.svozil@tuwien.ac.at}
\homepage{http://tph.tuwien.ac.at/~svozil}

\affiliation{Institute for Theoretical Physics,
TU Wien,
Wiedner Hauptstrasse 8-10/136,
1040 Vienna,  Austria}

\date{\today}

\begin{abstract}
Using a group-theoretic approach, a method for determining the equivalence classes
(also called orbits) of the set of rules of one-dimensional cellular automata induced by
the symmetry operations of reflection and permutation and their product is presented.
Orbits are classified by their isomorphism type.
Results for the number of orbits and the number of orbits by type for state sets
of size two and three are included.
\end{abstract}


\maketitle
\newpage

\section{Introduction}

\subsection{The Physical Relevance of Cellular Automata}

Cellular Automata (CAs) are mathematical constructs that model systems composed of discrete components evolving over time according to simple local rules.
Despite their simplicity, CAs exhibit remarkable phenomenological complexity, making them powerful tools for studying a wide range of natural and computational phenomena.
 Beyond their abstract utility, CAs hold profound potential physical relevance as models for discrete universes and simulations,
 offering insights into the fundamental principles underlying locally governed (though not necessarily spatially localized) universes and the dynamics of complex systems.
Their capacity for Church-Turing universal computation, including the self-reproduction of universal devices within their framework, provides metaphors that may extend to continuous physical models.

Historically, Konrad Zuse, in his seminal work \textit{Rechnender Raum} (Calculating Space)~\cite{zuse-69,zuse-70,zuse-94},
proposed the bold hypothesis that the universe itself could be interpreted as a vast computational structure evolving through local updates.
Zuse's digital physics posits that space, time, and matter are inherently discrete, with their evolution governed by computational rules analogous to those of CAs.
This perspective suggests that physical laws are emergent properties of an underlying computational substrate, where local interactions among discrete elements produce global patterns.
In this context, CAs serve as ideal candidates for modelling a  digital universe,  offering a conceptual framework for exploring the computational essence of reality.

Independently, John von Neumann, utilizing their algorithmic and computational aspects,  provided another profound perspective on the significance of CAs,
emphasizing their universality and self-replication capabilities~\cite{v-neumann-66}.
Motivated by questions of biological self-reproduction and universal computation, von Neumann designed a CA capable of replicating itself.
This achievement demonstrated that even within a simple, rule-based system, it is possible to encode the complexity of life-like processes and achieve computational universality.
Von Neumann's work laid the foundation for studying self-reproducing systems, influencing fields ranging from artificial life to nanotechnology,
and underscoring the potential of CAs to model the interplay between computation and dynamics.

CAs have also found extensive applications as models of dynamical systems.
Their discrete, rule-driven structure makes them particularly suited for simulating phenomena where local interactions give rise to emergent behaviour,
such as fluid dynamics, traffic flow, biological growth, and even aspects of quantum mechanics.
Unlike continuum-based models requiring analysis, CAs inherently capture the often granular, stepwise nature of many physical processes.

In sum, CAs embody a profound duality, functioning both as abstract computational models and as physically relevant systems.
From Zuse's vision of a computational universe to von Neumann's pioneering work on self-reproduction, CAs have reshaped our understanding of computation, biology, and the dynamics of physical systems.
This interplay between simplicity and complexity places CAs at the heart of efforts to unify computational theory with the physical world.

Their characterization in terms of equivalence through transformations such as reflection, state permutation, and their combinations represents a critical step toward understanding their potential.
This study is motivated by such physical motivations.

\subsection{Outline and Previous Work}

A one-dimensional CA operates on a bi-infinite lattice of cells where
each cell is in one state from a finite set of possible states.
A computational step of the automaton comprises the following operations.
For each cell the automaton reads the states of a small set of neighbouring cells including
the cell itself.
The values of the states read are used as input of a lookup table, called the local rule,
that determines the new state of the cell.
Then all cells are updated synchronously.
The net effect of one computational step is the calculation of a new bi-infinite sequence of states.

Multiple iterative computational steps of the CA leads to a sequence of configurations, termed the
evolution of the CA.
If each configuration of a CA's evolution is shifted the same number of cells to the left or to
the right, the CA's evolution is still governed by the same local rule.
This fundamental property of CA is called shift invariance.

Other symmetry operations transform the local rule.
If the CA's evolution is reflected (or mirrored),
the resulting evolution is governed by the reflected rule, which, in general,
is different, to the unmirrored one.
Reflection is thus a symmetry operation that transforms rules.

Similarly, since the states of a CA are merely labels,
permutating the labels does not change the dynamic behaviour of the CA, but will in general
result in a different local rule.
Rules that can be transformed into each other under reflection or permutation or
their product
are considered equivalent.
Consequently, the set of all CA rules splits up into classes of equivalent rules.

Wolfram \cite{Wolfram1983} designated the family of one-dimensional CAs
with two states and three neighbours as elementary.
In \cite{wolfram-1986}, pp.~485-557, he gave a table that divided the 256 rules of the
elementary CA into 88 equivalence classes with respect to the symmetry operations
of reflecting the lattice, permutation of the state set, and the product of these operations.
A mathematical derivation of this result
was carried out by Li and Packard \cite{Li1990TheSO}.
Cattaneo et al.~\cite{CATTANEO19971593} studied a variety of transformations of the set of local rules,
in particular, also the symmetry
operations of two-state CAs to be discussed in this work.
They gave, inter alias, the general result for the equivalence classes of two-state CAs with
$2r+1$ neighbours, where $r$ is a nonnegative integer.
The properties of symmetry transformations acting on CA rules have also been investigated, see
e.g., \cite{CASTILLORAMIREZ2020104533}.
Symmetry transformations were even extended to generalized CAs over groups,
see~\cite{Castillo-Ramirez03072023}.

Determining the equivalence classes is an elementary classification and serves both to understand
the set of local rules in terms of symmetry operations and
to reduce the number of non-equivalent rules.
There are a variety of other classification schemes.
For instance, Wolfram's classification~\cite{WOLFRAM19841} is based on the phenomenological
behaviour of the dynamic evolution,
the Culik-Yu classification \cite{10.5555/45269.45271} captures the computational complexity of the
limit sets;
see \cite{VISPOEL2022133074} for an overview.
The classification by symmetry operations precedes these higher-level classifications
as rules in an equivalence class are  all in the same class of other
classification schemes (at least they should be).

This study focuses on the equivalence classes of one-dimensional CAs induced by the symmetry transformations
of reflection and permutation and their product.
The set of symmetry transformations forms a group which acts on the set of CA rules.
Therefore group-theoretical concepts are applied to determine the equivalence classes of CAs.
In group-theoretical notation, equivalence classes induced by group actions are called orbits, and this
term is used in the following.
One of the main results of this study is the provision of formulas that give the number of orbits
for a state set of size two and three for any size of the neighbourhood.
This study goes beyond the scope of previous work by classifying the orbits according to their isomorphism type and
deriving the cardinalities of these classes.

The organisation of this study is as follows.
Section~\ref{definitions} provides definitions on CAs and symmetry operators.
Section \ref{sec:pre} presents a method based on group actions to determine the number of orbits.
Section~\ref{sec:two_states} and \ref{sec:three_states} contains the calculations for a state set of size two
and size three respectively.
Section~\ref{validation} presents a brute-force algorithm that can be used to validate the results for
small numbers of states and neighbours.
Conclusion remarks are made in Section~\ref{sec:summary}.

\section{Definitions}
\label{definitions}

\subsection{One-dimensional Cellular Automata}

The states of a CA are represented by symbols from a finite set, also called an alphabet.
As the symbols only serve to designate the states, any finite set will do, so we choose
the set $\Sigma=\{0,1,\ldots,k-1\}$  to represent a state set of size $k$.
The size (or cardinality) of an arbitrary set $A$ is denoted by $|A|$.

A word $w=x_0x_1\ldots x_{m-1}$ over an alphabet $\Sigma$ is a finite sequence of symbols
from $\Sigma$ juxtaposed.
The length of a word $w$, denoted $|w|$, is the length of the sequence,
that is $|x_0x_1\ldots x_{m-1}| = m$
(the notation $|.|$ denotes both the size of a set and the length of a word).
The set of all words of length $m$ over the alphabet $\Sigma$ is denoted by $\Sigma^m$.
A configuration $x$ is a bi-infinite sequence  over
the alphabet $\Sigma$, defined as a mapping of $\mathbb{Z}$ into $\Sigma$.
The $i$-th element, $i \in \mathbb{Z}$, of a configuration $x$ is denoted by $x_i$.

\begin{def2}
\label{def:CA}
A one-dimensional CA is a triple $(k, N, f)$, where \\
$k \geq 2$ is an integer, the number of states in the state set $\Sigma=\{0,1,\ldots,k-1\}$; \\
$N$ is the neighbourhood, a finite nonempty set of integers such that $-N = N + d$ for an integer $d$; \\
$f$ is the local rule, a function from $\Sigma^n$ to $\Sigma$. \\
Let $n = |N|$ and $N=\{j_0,j_1, \ldots,j_{n-1}\}$ such that $j_0 < j_1 < \cdots < j_{n-1}$.
The local mapping $f$ induces the global mapping on the set of configurations
$\Phi_f^N : \Sigma^\mathbb{Z} \rightarrow \Sigma^\mathbb{Z}$, defined by
$\Phi_f^N(x)_{i} = f(x_{i+j_0}x_{i+j_{1}}\ldots x_{i+j_{n-1}})$.
\end{def2}
We have used the notation $-N = \{-j \ | \ j \in N\}$ and $N + d = \{j + d \ | \ j \in N \}$.
If $N$ is given, we will write $\Phi_f$ instead of $\Phi_f^N$.
Def.~\ref{def:CA} is similar to the one used in \cite{castilloramirez2023},
 \cite{CULIK1990357} or \cite{KARI20053},
apart that we always use the first nonnegative integers as state set and more important
that we introduce the constraint $-N = N + d$ to later define the reflection operator in
a meaningful way. If the CA is initialised with the configuration $x$, the CA computes in one step
the configuration $\Phi_f(x)$.

The shift operator $\sigma$ operates on the set of configurations,
it shifts a configuration one cell to the left, formally defined by
$\sigma(x)_i = x_{i+1}$.
By the definition of the CA, the global mapping commutes with the shift operator:
$\Phi_f(\sigma(c)) = \sigma(\Phi_f(c))$.
A fundamental result of Hedlund \cite{DBLP:journals/mst/Hedlund69} shows that an alternative,
 topological definition
of an one-dimensional CA based on the shift operator and continuous mappings is equivalent to the one above.
We contrast Def.~\ref{def:CA} with another definition that is frequently found in literature,
e.g. \cite{JEN19903} or \cite{WOLFRAM19841}.
If $p$ and $q$ are  integers, let $[p,q]$ denote the integer interval $\{p, p+1, \ldots, q \}$.
\begin{def2}
\label{def:CA-r}
A radius-based CA is a CA $(k,N,f)$ such that $N = [-r,r]$, where
$r$ is a nonnegative integer, called the radius of the CA.
\end{def2}
The local mapping $f$ induces the global mapping $\Phi_f(x)_{i} = f(x_{i-r}x_{i-r+1}\ldots x_{i+r})$.
Def.~\ref{def:CA-r} encompasses only CAs with an odd ($2r+1$) number of neighbours.
The generalisation to an even number of neighbours becomes cumbersome, e.g. by shifting
the output configuration a half cell and introducing half-integers to index the configuration,
see Kari~\cite{KARI20053},  or by loss of symmetry, see Ruivo~\cite{RUIVO2018280},
while Def.~\ref{def:CA} enables a uniform treatment of all neighbourhood sizes.

If $k = |\Sigma|$, then the set $L(k,n) = \{f \, | \, f : \Sigma^n \to \Sigma \}$ is called the local rule space
of the family of CAs with $k$ states and $n$ neighbours.
The size of $L(k,n)$ is $k^{k^n}$.
If $k > 1$, $|L(k,n)|$ grows extremely fast as function of $n$:
$|L(k,0)| = k$, and $|L(k,n+1)| = |L(k,n)|^k$.
The set $G(k, N) = \{ \Phi_f \, | \, f \in L(k, |N|) \}$ is called the global rule
space of the family of CAs with $k$ states and neighbourhood $N$.
\begin{prop2}
If $G(k_1,N_1) = G(k_2,N_2)$, then $k_1 = k_2$ and $N_1 = N_2$.
\end{prop2}
\begin{proof}
If $k_1 \ne k_2$, then clearly $G(k_1,N_1) \ne G(k_2,N_2)$.
Suppose now that $k_1 = k_2$ and $N_1 \ne N_2$.
Then $N_1 \setminus N_2 \cup N_2 \setminus N_1$ is not empty.
Without loss of generality, suppose $N_1 = \{j_0, \ldots, j_{n-1}\}$
and $j_p \in N_1 \setminus N_2$.
Define a local rule $f$ by $f(a_0 \ldots a_{n-1})  = 1$ only if $a_p = 1$ and $a_i = 0$ for $i \ne p$,
and a configuration $x$ by $x_{j_p} = 1$ and $x_i = 0$ if $i \ne j_p$.
Then $\Phi_f^{N_1}(x)_i$ is $1$ if $i=0$ and $0$ otherwise.
Let $g \in L(k,n)$ be arbitrary.
If $\Phi_g^{N_2}(x)_0 = 0$ then $\Phi_f^{N_1} \ne \Phi_g^{N_2}$.
If $\Phi_g^{N_2}(x)_0 = 1$, we conclude that $g(0 \ldots 0) = 1$ and $\Phi_g^{N_2}(x)$ contains
infinitely many 1's, so also $\Phi_f^{N_1} \ne \Phi_g^{N_2}$.
Thus, we have shown that $\Phi_f^{N_1} \not\in G(k, N_2)$.
\end{proof}
Note the following two cases.
First, if $N_2 = N_1 + q$ for an integer $q$, then
$G(k, N_1) = \{ \sigma^q \Phi_f \, | \, \Phi_f \in G(k, N_2) \}$.
Second, if $N_2 \subseteq N_1$,
then $G(k, N_2) \subseteq G(k, N_1)$.

\subsection{Symmetry operations}
\label{sec:sym-operations}
The notion of the equivalence of one-dimensional CAs is based on two classes of symmetry operations:
permutations of the state set and reflection of the configuration.

Let $S_k$ be the symmetric group of degree $k$,
that is the set of all permutations of the set $\Sigma = \{0,1, \ldots, k-1\}$,
and suppose $\alpha \in S_k$.
If $a \in \Sigma$ we write the image of $a$ under $\alpha$ as product $\alpha a$.
The extension of $\alpha$ to words and configurations is defined by elementwise application.
If $w = a_0\ldots a_{n-1} \in \Sigma^n$ is a word,
set $\alpha w  = (\alpha a_0) \ldots (\alpha a_{n-1})$.
If $x$ is a configuration, set $(\alpha x)_i = \alpha (x_i)$.
Suppose $f$ is a local rule that maps $\Sigma^n$ to $\Sigma$.
The permutation operator $\hat{\alpha}$ is
defined by
$\hat{\alpha}f(w) = \alpha f ( \alpha^{-1} w)$ for all words $w \in \Sigma^n$.
It represents a transformation
of the set of local rules.
Note that the  ``hat'' on the operator is necessary, because $\alpha f$ and
$\hat{\alpha} f$ are distinct entities.
The first one is the composite function $\alpha \circ f$, whereas the second represents
the composite function $\alpha \circ f \circ \alpha^{-1}$.
If $\Phi_f$ is the induced global mapping of $f$, we define
$\hat{\alpha} \Phi_f$ similarly:
$\hat{\alpha} \Phi_f (x) = \alpha \Phi_f (\alpha^{-1} x)$ for all configurations $x$.
From
\begin{eqnarray*}
\left( \hat{\alpha} \Phi_f(x) \right)_{i}
& = & \left( \alpha \Phi_f (\alpha^{-1} x) \right)_{i}
= \alpha f \left(\alpha^{-1}(x_{i+j_0} \ldots x_{i+j_{n-1}}) \right) \\
& = & \hat{\alpha} f(x_{i+j_0} \ldots x_{i+j_{n-1}})
 =  \Phi_{\hat{\alpha} f} (x)_{i}
\end{eqnarray*}
follows $\hat{\alpha} \Phi_f =  \Phi_{\hat{\alpha} f}$.

The second type of operator is the reflection operator.
If $w = a_0\ldots a_{n-1} \in \Sigma^n$ is a word over $\Sigma$,
define $rw  = a_{n-1} \ldots a_0$.
Note that $ra = a$ for all $a \in \Sigma$.
If $x$ is a configuration, set
$(rx)_i = x_{-i}$.
The reflection operator $\hat{r}$ is defined by
$\hat{r} f(w) = f(rw)$ and
$\hat{r} \Phi_f(x) = r \Phi_f(r x)$.
Since $r$ is self-inverse, that is $r^{-1} = r$,
we can also write
$\hat{r} f(w) = rf(r^{-1}w)$,
making the notation consistent with the one of the permutation operator.

From
\begin{eqnarray*}
\left( \hat{r} \Phi_f (x) \right)_{i} & = &
\left( r \Phi_f(rx) \right)_{i} =
\left( \Phi_f(rx) \right)_{-i} = f\left( (rx)_{-i +j_0} \ldots (rx)_{-i+j_{n-1}} \right) \\
&  = & f\left( x_{i  - j_0} \ldots x_{i-j_{n-1}} \right) = f\left( x_{i  + j_{n-1} + d} \ldots x_{i + j_{0} +d} \right) \\
&  = & f\left( r(x_{i  + j_0 + d} \ldots x_{i + j_{n-1} +d}) \right)
= \hat{r}f\left(x_{i  + j_0 + d} \ldots x_{i + j_{n-1} +d}\right)
=  \left( \Phi_{\hat{r}f}(x) \right)_{i+d};
\end{eqnarray*}
we conclude that  $  \hat{r} \Phi_f = \sigma^d \Phi_{\hat{r} f}$.
If the CA complies with Def.~\ref{def:CA-r} the relation simplifies to
$\hat{r} \Phi_f = \Phi_{\hat{r}f}$.

We call $R = \{1, r\}$, the reflection group.
The direct product of $S_k$ and $R$, written as $S_kR$, is the group that contains all permutations,
the reflection and their products.
Suppose that $\hat{\alpha}$ and $\hat{\beta}$ are two operators.
Then
\[
\hat{\alpha} \hat{\beta} f (w)  = \alpha \hat{\beta} f(\alpha^{-1} w)
= \alpha \beta f( \beta^{-1} \alpha^{-1} w) = \widehat{\alpha \beta} f(w).
\]
The operators form a group that is in general isomorphic to $S_kR$,
but for $n=1$ (or $k=1$) the relation is only a homomorphism.
Note that the reflection operator commutes with all permutation operators.
If the global mapping $\Phi_f$ satisfies $\Phi_f = \hat{\alpha} \Phi_f$,
the CA is said to be invariant under the operator $\hat{\alpha}$.

The meaning of the operators defined above is illustrated by the following observation.
Suppose $\hat{\alpha}$ is
either one of the permutation operators or the reflection operator,
and consider two radius-based CAs (Def.~\ref{def:CA-r})
with the same state set and the same radius
and respectively, with local rule $f$ and local rule $\hat{\alpha} f$.
If the initial configuration of CA $A$ is $x$ and the one of CA $B$ is $\alpha x$,
then the same 1-1 correspondence between the configurations established by $\alpha$ persists
for all iterations:
$\alpha \Phi_f^t(x) = \Phi_{\hat{\alpha} f}^t (\alpha x)$ holds for any positive integer $t$
($\Phi^t$ denotes the $t$-th iteration of $\Phi$).
Suppose now that the CAs are of the general form of Def.~\ref{def:CA}.
The same relation holds, if $\alpha$ is a permutation, but if $\alpha = r$, it changes.
Then CA $A$ is after one step in configuration $\Phi_f(x)$, and CA $B$ in configuration
$\Phi_{\hat{r}f}(rx)$.
Using $\hat{r} \Phi_f = \sigma^d \Phi_{\hat{r} f}$, we obtain
$\Phi_{\hat{r}f}(rx) = \sigma^{-d} \hat{r} \Phi_f(rx) = \sigma^{-d} r \Phi_f(x)$,
so $ r \Phi_f(x) = \sigma^d \Phi_{\hat{r}f}(rx)$.
For any number $t$ of steps, the relation becomes
$ r \Phi_f^t(x) = \sigma^{dt} \Phi_{\hat{r}f}^t(rx)$.

\section{Preliminaries}
\label{sec:pre}

\subsection{Groups and Group Actions}

We assume some basic knowledge of groups as it can be found in introductory textbooks, e.g.
\cite{dummit2003abstract}, \cite{milneGT}, or \cite{rotman2012introduction}.
However, we briefly introduce the notation that is used in the following,
define group actions and related concepts and
state some propositions about them,
all of these to be found in more depth and more relaxed pace in the references above.

Let $H$ be a subgroup of $G$, denoted by $H \le G$.
If $g \in G$, the left coset of $H$ in $G$ is defined by
$gH = \{ gh \ | \ h \in H \}$.
The index $[G : H]$ of $H$ in $G$ denotes the number of left cosets of $H$ in $G$.
Lagrange's theorem states that $|G| = [G : H] \times |H|$.
If $g \in G$, the conjugate of $H$ by $g$ is the set
$gHg^{-1} = \{ ghg^{-1} \ | \ h \in H \}$,
which is also a subgroup isomorphic to $H$.
A subgroup $N$ of $G$ is called normal if $gNg^{-1} = N$ for all $g \in G$.

A group action of a group $G$ on a set $A$ is a map from $G \times A$ to $A$
satisfying the following properties:
\begin{enumerate}[(i)]
\item $g_1(g_2a) = (g_1g_2)a$ for all $g_1,g_2 \in G$, $a \in A$, and
\item $1a=a$, for all $a \in A$.
\end{enumerate}

The relation on $A$, defined by $a \sim b$ if and only if $a=gb$ for some $g \in G$, is an equivalence
relation.
The equivalence classes $[a] = \{ga : g \in G\}$ are called $G$-orbits (or just orbits),
and the set of orbits forms a partition of $A$, denoted by $ A / G$.
The length of an orbit $[a]$ is its size $|[a]|$.
An element $a \in A$ is fixed by $g \in G$ if $ga = a$.
The set of all group members that fix an element $a \in A$ is
called the stabilizer of $a$, that is the set
$\stb(a) = \{ g \in G \, | \, ga = a \}$,
which forms a subgroup of $G$.
If $g \in G$, the set of all fixed points of $g$ is denoted by
$\fix(g) = \{a \in A \, | \, ga = g \}$.
The notation is generalized to subgroups.
If $H \le G$, then the set
$\fix(H) = \{ a \in A \, | \, \mbox{$ga = a$ for all $g \in H$} \}
= \bigcap_{g \in H} \fix(g)$,
consists of all elements of $A$ that are fixed points for all $g \in H$.

\begin{prop2}[Orbit-Stabilizer Theorem]
If the group $G$ acts on $A$ and $a \in A$, then the length of the $G$-orbit which contains $a$ is
equal to the index of the stabilizer of $a$ in $G$:
\[
|[a]| = [G : \stb(a)].
\]
\end{prop2}

\begin{proof}
The map $ga \mapsto g\stb(a)$ that associates the element $ga$ of the orbit with the left coset
$g\stb(a)$ is well-defined and bijective.
\end{proof}

Every group G acts on the family of all its subgroups by conjugation.
The orbits of this group action are called conjugacy classes.
If $H \le G$, then the conjugacy class of $H$ is the set of subgroups $[H] =
\{ H^\prime \le G \, | \, \mbox{$H^\prime = gHg^{-1}$ for some $g \in G$} \}$.
The set of conjugacy classes is denoted by ${\cal C}(G)$.
If $H \le G$ and $H$ is normal, then the orbit containing $H$ is a singleton.
If $G$ is abelian, each orbit of ${\cal C}(G)$ is a singleton.
If $H_1$ and $H_2$ are subgroups of $G$, the relation $H_1 \le H_2$ is a partial order on the
set of subgroups.
It induces a partial order on ${\cal C}(G)$ by $[H_1] \leq [H_2]$ if and only if
there is a $H_1^\prime \in [H_1]$ and a $H_2^\prime \in [H_2]$ such that $H_1^\prime \leq H_2^\prime$.
The lattice $({\cal C}(G), \le)$ is called the reduced subgroup lattice of $G$.

We consider again a group $G$ acting on an (arbitrary) set $A$.
The orbit $O \in A/G$ is said to be of type $[H] \in {\cal C}(G)$ if the stabilizer
of some $a$ in $O$ belongs to $[H]$.
If two orbits $O_1$ and $O_2$ are of the same type, then there is a bijection
$\varphi: O_1 \to O_2$, such that $\varphi(ga) = g\varphi(a)$ for all $g \in G$ and all $a \in O_1$.
The function $\varphi$ is called a $G$-isomorphism.
Define
$
\orb(A/G, H) = \{ O \in A/G \, | \, \mbox{the type of $O$ is $[H]$} \}.
$
Note that $|A / G| = \sum_{[H] \in {\cal C}(G)} |\orb(A / G, H)|$.

Having established the terminology, we consider now the family of one-dimensional CAs with $k$ states and
$n$ neighbours.
The mapping $S_kR \times L(k,n) \rightarrow L(k,n)$;
$(\alpha,f) \mapsto \hat{\alpha}f$ fulfils the properties of a group action.
If $f \in L(k,n)$, the orbit of $f$
is the set $[f] = \{\hat{\alpha} f \ | \  \alpha \in S_kR \}$ and the
set of all orbits is denoted by $L(k,n) / S_kR$.
Local rules in the same orbit are connected by symmetry transformations, while
orbits of the same type cannot be distinguished by symmetry transformations alone.
We abbreviate $\orb( L(k,n) / S_kR, H )$ to $\orb(k,n,H)$.
The aim of this study is to develop a method for
determining $L(k,n) / S_kR$ and the sets $\orb(k,n,H)$ where $[H] \in {\cal C}(S_kR)$,
and in particular to derive formulas for the cardinalities of these sets.

\subsection{Counting Orbits}
\label{sec:count_orbs}

The following lemma relates the number of orbits to the number of fixed points of the group elements.

\begin{prop2}[Burnside's Lemma]
Let $G$ be a group acting on the set $A$. The number of $G-orbits$ is
\[
|A/G| = \frac{1}{|G|} \sum_{g \in G} |\fix(g)|.
\]
\end{prop2}

\begin{proof}
In  the sum $\sum_{g \in G} |\fix(g)|$, each $a \in A$ is counted $|\stb(a)|$ times
(for $\stb(a)$ consists of all those $g \in G$ which fix $a$).
If $a$ and $b$ lie in the same orbit, then $b=ga$ for a $g \in G$.
This implies $\stb(b) = g \stb(a) g^{-1}$, and in particular $|\stb(b)| = |\stb(a)|$.
So, the $[G : \stb(a)]$ elements constituting the orbit of $a$ are, in the above sum,
collectively counted $[G : \stb(a)] \times |\stb(a)|$ times.
Each orbit thus contributes $|G|$ to the sum, and so
$\sum_{g \in G} \fix(g) = |A/G| \times |G|$.
\end{proof}

The proof was adapted from \cite{rotman2012introduction}.
Burnside's lemma gives the total number of orbits.
Since we are also interested in the distribution of orbits by type,
we will use the following method in Section~\ref{sec:two_states} and Section~\ref{sec:three_states}.
Let $[H] \in {\cal C}(G)$.
The set $\stb^{-1}([H]) = \bigcup_{H^\prime \in [H]} \stb^{-1}(H^\prime)$
is the union of all orbits of type $[H]$, all having length $[G : H] = |G| \, / \, |H|$.
Thus
\begin{equation}
\label{eq:orb_H_card}
|\orb(A/G, H)| = |\stb^{-1}([H])| \, / \, [G : H] = |\stb^{-1}(H)| \times |[H]| \times |H| \, / \, |G|.
\end{equation}
In calculating the numbers $|\stb^{-1}(H)|$ we take a detour.
The mapping $\fix(H)$ from the set of subgroups of $G$ into $A$ does not create a partition of $A$:
if $H_1$ is a proper subgroup of $H_2$ ($H_1$ is a subgroup of $H_2$ and $H_1 \ne H_2$),
denoted by $H_1 < H_2$, then
$\fix(H_2) \subset \fix(H_1)$.
The mappings $\fix$ and $\stb$ are related
\begin{equation}
\label{eq:stab_inv_set}
\stb^{-1}(H) = \fix(H) \setminus \bigcup_{H < H^\prime} \stb^{-1}(H^\prime),
\end{equation}
where $H^\prime$ is also assumed to be a subgroup of $G$.
Since the sets $\stb^{-1}(H)$ are disjoint, Equation~(\ref{eq:stab_inv_set}) implies
\begin{equation}
|\stb^{-1}(H)| = |\fix(H)| -
 \sum_{H < H^\prime} |\stb^{-1}(H^\prime)|.
\label{stab_inv}
\end{equation}
To calculate $|\stb^{-1}(H)|$ for all subgroups, we start with $G$, for which
$\stb^{-1}(G) = \fix(G)$ holds, the calculation of the subgroups can then
be done successively.

In general, the numbers $|G|$ and $|\fix(g)|$ are not sufficient to determine the
numbers $|\orb(A / G, H)|$.
The following example describes two different group actions of the same group on the same set, so
that the numbers $|\fix(g)|$ are the same, but the distribution of orbits by type is different.
If a group $G$ acts on a set $A$, it induces a homomorphism $\varphi: G \to S_A$;
$g \mapsto (a \mapsto ga)$, where
$S_A$ denotes the symmetric group of $A$.
We can therefore associate group elements of $G$ with permutations of the set $A$.
Let $V = \{1,a,b,c\}$ the Klein four-group and $A =\{1,2,\ldots,6\}$.
Consider the group actions $\psi_1$ and $\psi_2$, both mappings of $V \times A$ onto $A$,
where $\psi_1$ has the permutation representation
\[
\sigma_1 = (), \sigma_a = (12)(34), \sigma_b = (34)(56), \sigma_c = (12)(56);
\]
and $\psi_2$ is given by
\[
\tau_1 = (), \tau_a = (12)(34), \tau_b = (13)(24), \tau_c = (14)(23).
\]
It is easily verified that these representations satisfy the group axioms and are isomorphic to $V$.
The first group action $\psi_1$ partitions $A$ into the orbits
$\{1,2\}$, $\{3,4\}$, and $\{5,6\}$, the second group action $\psi_2$ leads to the partition
$\{1,2,3,4\}, \{5\}$, and $\{6\}$.
Note that  $|\fix(a)| = |\fix(b)| = |\fix(c)| = 2$
and
$|A / V| = 3$ for both actions,
while, for instance,
$\orb(A / V, V) = \emptyset$ for the first action, but $\orb(A / V, V) = \{\{5\},\{6\}\}$ for the second one.

\subsection{Symmetry Operators acting on the Set of Local Rules}
\label{sec:sym_op_domain}

The domain of a local rule is the set $\Sigma^n$ of all words over $\Sigma$ having length $n$.
If $H$ is a subgroup of $S_kR$, the mapping $H \times \Sigma^n \rightarrow \Sigma^n$ defined
by $(\alpha, w) \mapsto \alpha w$ satisfies the properties of a group action.

We will now study mappings that are defined on an orbit of $\Sigma^n / H$.
Suppose $A \subset \Sigma^n$ is an $H$-orbit, and $g$ is a mapping $A \to \Sigma$.
If $\alpha \in H$ then $\alpha A = \{\alpha w | w \in A\} = A$.
This shows that the domain of $\hat{\alpha}g$ is also $A$.
Hence we can speak about functions defined on $A$ that are invariant under $H$.

The set $\{A_1, \ldots, A_p\}$ of all $H$-orbits is a  partition of $\Sigma^n$.
If $f$ is a local rule invariant under $H$,
then the restriction $f|A_i$ is clearly also invariant under $H$.
On the other hand, if $g_i: A_i \rightarrow \Sigma$; $i = 1 \ldots, p$,
is a sequence of mappings,
all invariant under $H$, then the local rule defined by
$f(w) = g_i(w)$ if $w \in A_i$ is also invariant under $H$.
This shows that invariant functions defined on the orbits are the building
blocks of invariant functions defined on $\Sigma^n$.

Let $A$ be again an $H$-orbit $A$  and suppose
$g: A \rightarrow \Sigma$ is invariant under $H$.
Choose a word $w$ in $A$, and  consider a different word in $A$, say $v$.
Since there is an $\alpha \in H$ such that $v = \alpha w$, the relation
$g(v) =  \hat{\alpha} g(v) = \alpha g(\alpha^{-1} v) = \alpha g(w)$ holds, and
the value of $g(v)$ is determined by $g(w)$.
This implies that there are at most $k$ different mappings $g: A \rightarrow \Sigma$
that are invariant under $H$.

\subsection{Examples}

We will study the group action of $\langle (01)r \rangle$ on some of the orbits of
two-state and three-state neighbourhoods.

\begin{enumerate}

\item
\label{example_k2_n_even}
Let $\Sigma = \{0,1\}$, and $n=2m$ be a positive even integer.
Suppose that the group $\langle  (01)r \rangle$ acts on $\Sigma^n$.
Consider the word $w=0^m1^m$ ($m$ copies of 0 followed by $m$ copies of 1).
The group action of $(01)r$ on $w$ results in
\[
(01)rw = (01)r(0^m1^m) = (01)(1^m0^m) = 0^m1^m = w;
\]
and so the set $ A = \{ w \}$ represents a singleton orbit.
Assume there is a function $f$ from $A$ to $\Sigma$ that is invariant under $\langle (01)r \rangle$.
Then $f$ has to satisfy the relation $f((01)rw) = (01)rf(w)$.
But since $f((01)rw) = f(w)$, we obtain the contradiction $f(w) = (01)f(w)$.
This shows that there is no local rule on $\Sigma^n$ that is invariant under $\langle (01) r \rangle$.

\item
\label{k2n_even}
Let $\Sigma$ be as above, let $n=2m+1$ be a positive odd integer, and let $w \in \Sigma^n$.
If we write $w = ucw$, where $u$ and $v$ are words of length $m$ and $c$ is a symbol of $\Sigma$,
we see that
\[
(01)rw = (01)r(ucv) = (01)(rv)c(ru) = ((01)rv)((01)c)((01)ru).
\]
Since for all $c \in \{0,1\}$, $c \ne (01)c$, we conclude that $w \ne (01)rw$, and
that $A = \{w, (01)rw \}$ is an orbit of length 2.
Choose a symbol $a$ from $\{ 0,1 \}$ and set $f(w) = a$.
If we set $f((01)rw) = (01)a$, the function $f$ is invariant under $\langle (01)r \rangle$.
Since $a$ was arbitrary, there are two functions with domain $A$ that are invariant under
$\langle (01)r \rangle$.

\item
Let $\Sigma = \{0,1,2\}$, $n=2m$ be a positive even integer, and $w=0^m1^m$.
The singleton $A = \{w\}$ is an orbit of $\langle (01)r \rangle$.
Assume that $f$ is invariant on $A$.
Then $f(w) = (01) f(w)$ must hold,
which is satisfiable by the choice $f(w) = 2$.
Hence there exists exactly one function from $A$ to $\Sigma$ that is invariant under
$\langle (01)r \rangle$.

\item
\label{k=3,n=2m+1}
Let $\Sigma$ be as above, but let $n=2m+1$ be a positive odd integer.
Consider a word $w$ of $\Sigma^n$ and write it in the form $w=ucv$,
where $u$ and $v$ are words of length $m$ and $c$ is a symbol.
The relation $w  = (01)r w$ leads to the constraints $v = (01)r u$ and $c=2$,
satisfied by $3^m$ words.
If $w$ is one of these words, a function $f$ defined on the orbit $\{w\}$ that is invariant is constrained
to the value $f(w) = 2$.
All other orbits of $\Sigma^n$ have length 2 and allow for three different invariant functions.

\end{enumerate}

\subsection{The Degree of an Orbit}

We have seen in Subsection~\ref{sec:sym_op_domain} that the number of invariant functions
on an orbit is at
most the size of the state set $k = |\Sigma|$.
The examples above have shown that the number of invariant functions
might also be smaller than $k$.
Let $H \le S_kR$ and $A$ be an orbit of $\Sigma^n / H$.
The degree of $A$ is defined to be the number of invariant functions on $A$, formally
\[
\deg(A) = | \{ f : A \to \Sigma \ | \ \mbox{ $\hat{\alpha} f = f$ for all $\alpha \in H$} \}|.
\]

The procedure for calculating the orbits of a CA which we will present shortly,
requires to determine the degree of a given orbit.
The following two lemmas will facilitate this task.
The first lemma states that to determine the degree of an orbit, it is sufficient to consider
all group actions on only
one word of the orbit.
The second lemma says that an orbit has degree $k$
 if the length of the orbit equals the order of the group.

\begin{prop}
\label{prop_inv}
Let $H \le S_kR$, $A$ be an $H$-orbit, $f$ be a function from $A$ to $\Sigma$, and
$w$ be any word of $A$.
If $f(\alpha w) = \alpha f(w)$ holds for all $\alpha$ in $H$,
then $f$ is invariant under $H$.
\end{prop}

\begin{proof}
By definition, $f$ is invariant under $H$ if $f(w) = \alpha f(\alpha^{-1} w)$ for
all $w \in A$ and $\alpha \in H$.
If we replace $\alpha$ by its inverse $\alpha^{-1}$, we see that the condition becomes
equivalent to $f(w) = \alpha^{-1} f(\alpha w)$ or $f(\alpha w) = \alpha f(w)$ for all
$w \in A$ and $\alpha \in H$.

Let $v \in A$ and $\beta \in H$.
Assuming that the condition of the lemma is fulfilled, we have to show that
$f(\beta v) = \beta f(v)$.
The proof is almost trivial.
Since $A$ is an $H$-orbit, there is a $\gamma \in H$, such that $v = \gamma w$.
Then $f(\beta v) = f((\beta \gamma) w) = (\beta \gamma) f(w)
= \beta (\gamma f(w)) = \beta f(\gamma w) = \beta f(v)$.
\end{proof}

\begin{prop}
\label{freeorb}
Let $H \le S_kR$ and $A$ be an $H$-orbit.
If $|A| = |H|$, then $\deg(A) = k$.
\end{prop}

\begin{proof}
Let $w$ be any word of $A$, and $a$ be any symbol of $\Sigma$.
Define a function $f$ from $A$ to $\Sigma$ as follows.
Set $f(w) = a$ and
let $v$ be another word of $A$.
Since $|A|=|H|$ there is exactly one $\alpha \in H$ such that $v = \alpha w$.
This shows that the value $f(v) = \alpha f(w) = \alpha a$ is well defined.
Lemma~\ref{prop_inv} implies that $f$ is invariant under $H$.
\end{proof}

\subsection{Outline of the Complete Calculation}
\label{subsec:outline_calc}

Given a local rule space $L(k,n)$,
the method to calculate the numbers $|\orb(k,n,H)|$, $[H] \in {\cal C}(S_kR)$ is as follows.

\begin{enumerate}

\item Construct the group $S_kR$ that is the direct product of the permutation group $S_k$ and
the reflective group $R = \{0, r\}$.
Having done that, construct the sublattice of all subgroups of $S_kR$.

\item For each conjugation class $[H]$ of ${\cal C}(S_kR)$ choose one representative $H \in [H]$.
Determine the orbits of
the group action $H \times \Sigma^n \to \Sigma^n$ and their degree.
Then
\[
|\fix(H)| = \prod_{A \in \Sigma^n / H} \deg(A).
\]

\item Beginning with $S_kR$ calculate $\stb^{-1}(H)$ for all selected representatives
by using Equation~(\ref{stab_inv}):
\[
|\stb^{-1}(H)| = |\fix(H)| -
 \sum_{H < H^\prime} |\stb^{-1}(H^\prime)|
\]

\item The number of orbits of type $[H]$ is given by Equation~(\ref{eq:orb_H_card}):
\[
|\orb(k, n, H)| = |\stb^{-1}(H)| \times |[H]| \times |H| \, / \, |G|.
\]

Sum up the numbers to obtain $|L(k,n) / S_kR|$, the total number of orbits (or apply
Burnside's lemma).

\end{enumerate}

\subsection{Shift-Equivalence}
\label{sec:shift-eq}
Another elementary equivalence relation was introduced in \cite{RUIVO2018280}, which we first illustrate
with an example in the domain of elementary CAs ($k=2$, $n=3$),
using Wolfram's nomenclature to label the rules,
see~\cite{WOLFRAM19841}.
Set $\Sigma=\{0,1\}$ and consider the elementary CAs $f_{12}$ and $f_{34}$, defined by
\[
f_{12}(w) =
\left\{
\begin{array}{ll}
1 & \mbox{if $w=010$ or $w=011$}  \\
0 & \mbox{otherwise}
\end{array}
\right.
\ \ \mbox{and} \ \
f_{34}(w) =
\left\{
\begin{array}{ll}
1 & \mbox{if $w=001$ or $w=101$}  \\
0 & \mbox{otherwise}
\end{array}
\right.
\]
Define a function $h: \Sigma^2 \to \Sigma$ by $h(01) = 1$ and $h(w) = 0$ if $w \ne 01$.
Then $f_{12}(a_0a_1a_2) = h(a_0a_1)$ and $f_{34}(a_0a_1a_2) = h(a_1a_2)$
for all $a_0a_1a_2 \in \Sigma^3$.
It is easy to see that $\Phi_{f_{34}} = \sigma \Phi_{f_{12}}$.

Two CAs $(k, N_1, f)$ and $(k, N_2, g)$ are said to be shift-equivalent if
$\Phi_f^{N_1} = \sigma^j \Phi_g^{N_2}$ for some integer $j$, denoted by
$\Phi_f^{N_1} \stackrel{\sigma}{\sim} \Phi_g^{N_2}$.
The relation $\stackrel{\sigma}{\sim}$ is an equivalence relation.
A mapping $p: \Sigma^n \to \Sigma^m$, $m \le n$, is called a projection if
there are integers $q_0, \ldots, q_{m-1}$ such that $0 \leq q_0 < \ldots < q_{m-1} \leq n-1$ and
$p(a_0\ldots a_{n-1}) = a_{q_0}\ldots a_{q_{m-1}}$ for all $a_0\ldots a_{n-1} \in \Sigma^n$.
The set $\{q_0, \ldots, q_{m-1}\}$ is called the index set of the projection.
A local rule $f: \Sigma^n \to \Sigma$ is called reducible if there
exists a local rule $h: \Sigma^m \to \Sigma$ and a projection $p: \Sigma^n \to \Sigma^m$ with $m < n$,
such that $f = h \circ p$, otherwise $f$ is said to be irreducible.
Let $N = \{j_0, \ldots, j_{n-1} \}$ be the neighbourhood of the CA.
If $f$ is reducible, then there is a rule $h$ and a projection $p$ such that $f = h \circ p$,
and the index set of $p$ is minimal.
If $\{q_0, \ldots, q_{m-1}\}$ is this index set,
the set $M = \{j_{q_0}, \ldots, j_{q_{m-1}} \} \subset N$ is called the support of $f$.

Suppose that $f$ is reducible with support $M$, $f = h \circ p$, the index set of $p$ is
$\{q_0, \ldots, q_{m-1}\}$,
 and there is an integer $t$ such that $M^\prime = M + t \subset N$.
Write $M^\prime = \{j_{q_0^\prime}, \ldots, j_{q_{m-1}^\prime} \}$.
Let $p^\prime$ be the projection $\Sigma^n \to \Sigma^m$ with index set
$\{q_0^\prime, \ldots, q_{m-1}^\prime\}$
and define a local rule with the same neighbourhood $N$ as $f$ by $f^\prime = h \circ p^\prime$.
From
\begin{eqnarray*}
\Phi_{f^\prime}(x)_i
& = & (h \circ p^\prime)(x_{i+{j_0}} \ldots x_{i+j_{n-1}})
= h(x_{i+j_{q_0^\prime}} \ldots x_{i+j_{q_{m-1}^\prime}})
=  h(x_{i+j_{q_0}+t} \ldots x_{i+j_{q_{m-1}}+t}) \\
& = & \sigma^t h(x_{i+j_{q_0}} \ldots x_{i+j_{q_{m-1}}})
=  \sigma^t (h \circ p)(x_{i+{j_0}} \ldots x_{i+{j_{n-1}}})
= \left( \sigma^t \Phi_f (x) \right)_i
\end{eqnarray*}
follows $\Phi_{f^\prime} = \sigma^t \Phi_f$, so $f$ and $f^\prime$ are shift-equivalent.

Let $N_1$ and $N_2$ be integer intervals.
Suppose $|N_1| = n$, $|N_2| = m$, $m \le n$, $f \in L(k, m)$, and $f$ is irreducible.
Then there are $n-m+1$ local rules $g_i \in L(k,n)$ such that
$\Phi^{N_1}_f \stackrel{\sigma}{\sim} \Phi^{N_2}{g_i}$.
The reflected rules $\Phi^{N_2}_{\hat{r}g_i}$ are all different, but shift-equivalent:
$\Phi^{N_1}_{\hat{r}f} \stackrel{\sigma}{\sim} \Phi^{N_2}_{\hat{r}g_i}$.
If $f$ is reflection-symmetric, $\hat{r}f = f$,
then, if at all, only one of the $g_i$, is reflection-symmetric.
Note also the general relation $\hat{r} \Phi_f \stackrel{\sigma}{\sim} \Phi_{\hat{r} f}$
from Subsection~\ref{sec:sym-operations}.
There is the special case of $|N_1| = 2$, and $|N_1|$ = 1.
Let $f \in L(k,1)$, and consider the rules $g_0, g_1 \in L(k,2)$, such
that $g_0(a_0a_1) = f(a_0)$ and $g_1(a_0a_1)$ = $f(a_1)$.
Here, equivalence by reflection and by shift coincide:
$g_1 =  \hat{r} g_0$ and $\Phi_{g_1} \stackrel{\sigma}{\sim}  \Phi_{g_0}$.
This shows that the orbits of $L(k, 2) / S_kR$ remain unchanged if shift-equivalence
is taken into account.

Let $N_1 = [-r_1,r_1]$, and $N_2 = [-r_2,r_2]$ with $r_2 < r_1$.
If $f \in L(k, 2r_2+1)$, $p(a_{-r_1} \ldots a_{r_1}) = a_{-r_2} \ldots a_{r_2}$, and
$g = f \circ p \in L(k, 2r_1+1)$, then $\Phi_g = \Phi_f$.
In this case, $\Phi_{\hat{r} g} = \Phi_{\hat{r} f}$ holds.

A local rule defined on a smaller neighbourhood might reappear in multiple copies that are shift-equivalent
when considering  larger neighbourhoods.
We conclude that shift-equivalence is another important elementary equivalence relation.
However, it is of a different nature since it concerns only local rules that can be defined on a proper subset
of the neighbourhood.
The consideration  of shift-equivalence into the presented framework, which is based on
the group of
permutations and reflection, goes beyond the scope of this study.
Results  for the number of equivalence classes,
obtained with a computer program
that also considered shift-equivalence,
can be found in \cite{RUIVO2018280} for small families of CAs ($k=2$, $n=2,3,4$ and
$k=3$, $n=2$).

\section{Two States}
\label{sec:two_states}

This section deals with the orbits of of one-dimensional two-state CAs,
that is $\Sigma = \{0,1\}$.

\subsection{The Group $S_2R$}

$S_2$, the symmetric group of degree 2, contains as its elements the identity $1$ and the transposition $(01)$.
The direct product $S_2R$ is given by
\[
S_2R = \langle (01) \rangle \langle r \rangle = \{1, (01) \}\{1 , r \}
= \{ 1, (01), r, (01)r \} = \langle (01), r \rangle.
\]
The notation $\langle \alpha, \beta, \ldots \rangle$ is called a generator,
and denotes the group with the property that
every element of the group can be written as finite product of the elements of the generator and
their inverses.
Since $S_2$ is abelian and $r$ commutes with $1$ and $(01)$, $S_2R$ is abelian too.
It is isomorphic to the Klein four-group.
The lattice of subgroups is depicted in Fig.~\ref{fig:lattice_s2r}.

\begin{figure}
\centering
\includegraphics[scale=0.7]{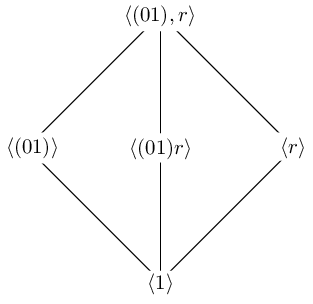}
\caption{The lattice of $S_2R$.}
\label{fig:lattice_s2r}
\end{figure}

\subsection{Odd Number of Neighbours}

Suppose the number of neighbours is odd, $n = 2m + 1$, $m=0,1,2,\ldots$.
For each subgroup $H$ of $S_2R$, we will calculate the number of orbits of $H$ acting on $\Sigma^n$.
For the next paragraphs $w$ denotes a word of $\Sigma^n$.

\begin{enumerate}

\item The group $\langle 1 \rangle$.
All orbits are of length 1, and so $\Sigma^n$ partitions into $2^{2m+1}$ singletons.

\item The group $\langle r \rangle$.
The relation $w = rw$ is satisfied if and only if the word $w$ is of the form $ua(ru)$ where
$|u| = m$ and $a$ is a symbol.
This relation is fulfilled by $2^{m+1}$ words.
The set of these words divides into $2^{m+1}$ orbits of length 1.
The remaining $2^{2m+1} - 2^{m+1}$ words of $\Sigma^n$ are all in orbits of length 2.
The total number of orbits is therefore
$(2^{2m+1} - 2^{m+1}) / 2 + 2^{m+1} = 2^m(2^m +1 )$.

\item The group $\langle (01) \rangle$.
Consider any word $w$ of $\Sigma^n$ and let $a$ denote the symbol in the centre of the word.
Then $(01)a$ is the symbol in the centre of the word $(01)w$.
This shows that the words $w$ and $(01)w$ are always different.
Hence $\Sigma^n$ partitions into $2^{2m}$ orbits of length 2.

\item The group $\langle (01)r \rangle$.
As before, the words $w$ and $(01)rw$ are always different.
Hence the number of orbits is again $2^{2m}$.

\item The group $S_2R$.
If $w=rw$ holds, then also $r(01)w = (01)w$ holds.
We have seen that there are $2^{m+1}$ words satisfying the relation $w = rw$, and so there
are $2^{m}$ orbits of length 2 consisting of the words $w$ and $(01)w$.
The remaining $2^{2m+1} - 2^{m+1}$ words divide into orbits of length 4.
Thus the total number of orbits is
$(2^{2m+1} - 2^{m+1})/4 + 2^{m} = 2^{m-1}(2^m+1)$.

\end{enumerate}
All the orbits of the groups have degree 2.
Put $H_1 = S_2R$, $H_2 = \langle (01) \rangle$,
$H_3 = \langle (01)r \rangle$, $H_4 = \langle r \rangle$ and $H_5 = \langle 1 \rangle$.
Set $a_i =|\fix(H_i)|$, $b_i = |\stb^{-1}(H_i)|$, and
$c_i = |\orb(2, n, H_i)|$, for $i=1, \ldots, 5$.
If $p_i$ is the number of orbits (of degree 2) of the group $H_i$, then
$a_i = 2^{p_i}$.
We get
\begin{eqnarray*}
& & b_1 = a_1 = c_1; \\
& & b_2 = a_2 - a_1, c_2 = b_2 / 2; \\
& & b_3 = a_3 - a_1, c_3 = b_3 / 2; \\
& & b_4 = a_4 - a_1, c_4 = b_4 / 2; \\
& & b_5 = a_5 - a_1 - b_2 - b_3 - b_4 = a_5 + 2 a_1 - a_2 - a_3 - a_4,
    c_5 = b_5 / 4.
\end{eqnarray*}
The total number of orbits is given by $\sum_i c_i$.
Expressing the $c_i$'s by the $p_i$'s leads to the following result.

\begin{prop2}
\label{prop:S2R_odd}
Let $m$ be a nonnegative integer.
\begin{enumerate}[(i)]
\item The set of rules $L(2, 2m+1)$ partitions into
\label{prop:k3_orbits}
$
\frac{1}{4} \left( 2 \times 2^{2^{2m}} + 2^{2^{m}(2^m+1)} + 2^{2^{2m+1}} \right)
$
orbits.
\item
The number of orbits of $L(2, 2m+1)$ by type are
\[
\begin{array}{lll}
|\orb(2, 2m+1, S_2R)| & = & 2^{2^{m-1}(2^m+1)};\\
|\orb(2, 2m+1,\langle (01) \rangle)| & = & \frac{1}{2}\left( 2^{2^{2m}} - 2^{2^{m-1}(2^m+1)} \right); \\
|\orb(2, 2m+1,\langle (01)r \rangle)| & = &  \frac{1}{2}\left( 2^{2^{2m}} - 2^{2^{m-1}(2^m+1)} \right);  \\
|\orb(2, 2m+1,\langle r \rangle)| &  = & \frac{1}{2}\left( 2^{2^m(2^m+1)} - 2^{2^{m-1}(2^m+1)} \right);  \\
|\orb(2, 2m+1,\langle 1 \rangle)| &   = & \frac{1}{4}\left(
2^{2^{2m+1}} + 2 \times 2^{2^{m-1}(2^m+1)} - 2 \times 2^{2^{2m}} - 2^{2^m(2^m+1)}
\right).
\end{array}
\]
\end{enumerate}
\end{prop2}
Part~(\ref{prop:k3_orbits}) is a particular case of Proposition 21 in \cite{CATTANEO19971593}.

\subsection{Even Number of Neighbours}

Suppose the number of neighbours is even, $n = 2m$, $m=1,2,\ldots$.
For the next paragraphs $w$ denotes a word of $\Sigma^n$.
Example~\ref{example_k2_n_even} has shown that there is no local rule invariant
under $\langle (01)r \rangle$, and by implication no local rule invariant under $S_2R$.
Thus, in calculating the number of orbits, we therefore only have to consider the remaining subgroups.

\begin{enumerate}

\item The group $\langle 1 \rangle$.
The set $\Sigma^n$ partitions into $2^{2m}$ orbits of length 1.

\item The group $\langle (01) \rangle$.
The set $\Sigma^n$ partitions into $2^{2m-1}$ orbits of length 2.

\item The group $\langle r \rangle$.
The set $\Sigma^n$ partitions into $2^{m-1}(2^m+1)$ orbits.

\end{enumerate}

All the orbits of the three groups above are of degree 2.
Put $H_1 = \langle r \rangle$, $H_2 = \langle (01) \rangle$, and
$H_3 = \langle 1 \rangle$.
Set $a_i =|\fix(H_i)|$, $b_i = |\stb^{-1}(H_i)|$, and
$c_i = |\orb(2,n,H_i)|$, for $i=1, 2, 3$.
If $p_i$ is the number of orbits (of degree 2) of the group $H_i$, then
$a_i = 2^{p_i}$.
We get
\begin{eqnarray*}
& & b_1 = a_1, c_1 = b_1 / 2; \\
& & b_2 = a_2, c_2 = b_2 / 2; \\
& & b_3 = a_3 - b_1 - b_2, c_3 = b_3 / 4.
\end{eqnarray*}
The total number of orbits is given by $\sum_i c_i$.
Expressing the $c_i$'s by the $p_i$'s leads to the following result.
\begin{prop2}
\label{prop:S2R_even}
Let $m$ be a positive integer.
\begin{enumerate}[(i)]
\item
The set $L(2, 2m)$ of local rules partitions into
$
\frac{1}{4} \left(2^{2^{m-1}(2^m+1)} + 2^{2^{2m-1}} + 2^{2^{2m}} \right)
$
orbits.
\item
The number of orbits of $L(2, 2m)$ by type are
\[
\begin{array}{lll}
|\orb(2, 2m,\langle (01) \rangle)| & = & \frac{1}{2} 2^{2^{2m-1}} ;  \\
|\orb(2, 2m,\langle r \rangle)| &  = & \frac{1}{2} 2^{2^{m-1}(2^m+1)} ;  \\
|\orb(2, 2m,\langle 1 \rangle)| &  = &  \frac{1}{4}\left(
 2^{2^{2m}} - 2^{2^{2m-1}} - 2^{2^{m-1}(2^m+1)}
\right).
\end{array}
\]
\end{enumerate}
\end{prop2}

\begin{table}
\begin{center}
\caption{Count of two-state orbits by type.}
\label{tab:s2}
\begin{tabular}{|l|r|r|r|r|r|}
\hline
       & \multicolumn{5}{|c|}{$|\orb(2,n, H)|$} \\ \hline
$H$    & $n=1$ & $n=2$ & $n=3$ & $n=$4 & $n=5$ \\ \hline
$\langle (01), r \rangle$ & 2 & 0 &  8 & 0 & 1\,024 \\ \hline
$\langle (01)r \rangle$ 	&  0 & 0 & 4 & 0 & 32\,256 \\ \hline
$\langle (01) \rangle$ 	&  0 & 2 & 4 & 128 & 32\,256 \\ \hline
$\langle r \rangle$     &  1& 4 & 28 & 512 & 523\,776 \\ \hline
$\langle 1 \rangle$ 		&  0 & 1  & 44 & 16\,064 & 1\,073\,447\,424 \\ \hline
							&  3 & 7  & 88 & 16\,704 & 1\,074\,036\,736 \\ \hline
\end{tabular}
\end{center}
\end{table}

Table \ref{tab:s2} depicts the number of orbits of two-state CAs for
a neighbourhood size $n=1, \ldots,5$.
For each $n$ and each subgroup $H$ of $S_2R$ the table gives the number of orbits
of $\stb^{-1}(H)$.
The last row lists the total number of orbits for a given $n$.

\section{Three States}
\label{sec:three_states}

After calculating the orbits of one-dimensional two-state CAs in
Section~\ref{sec:two_states} this section deals with one-dimensional
three-state CAs, that is $\Sigma = \{0,1,2\}$.

\subsection{The Group $S_3R$}
The group $S_3R$, which is the direct product of the symmetric group $S_3$
and the reflection group $R$, contains all the symmetry operators of
one-dimensional three-state CAs.

Some remarks:
\begin{enumerate}
\item $S_3$ is not abelian, and neither is $S_3R$, for instance
$(01)(02) = (021)$, but $(02)(01) = (012)$.
\item The depiction of the lattice of subgroups, Fig.~\ref{fig:S3R_lattice},
 arranges groups of equal order in the same row.
From bottom to top the orders are 1, 2, 3, 4, 6, and 12.
\item The dashed rectangles demarcate conjugacy classes of subgroups.
The other subgroups are normal and form singleton classes with respect to conjugation.
Collapsing the conjugacy classes into single nodes result in the reduced subgroup lattice.
\item Groups are specified by generators, e.g. $S_3 = \langle (01), (12) \rangle$.
\item From $[(012)r]^3 = r$ follows $\langle (012)r \rangle = \langle (012), r \rangle$.
\item The group $\langle (01)r, (012) \rangle$ consists of the elements $1, (01)r, (12)r, (20)r, (012)$,
and $(021)$.
\end{enumerate}

\begin{figure}
\centering
\includegraphics[scale=0.7]{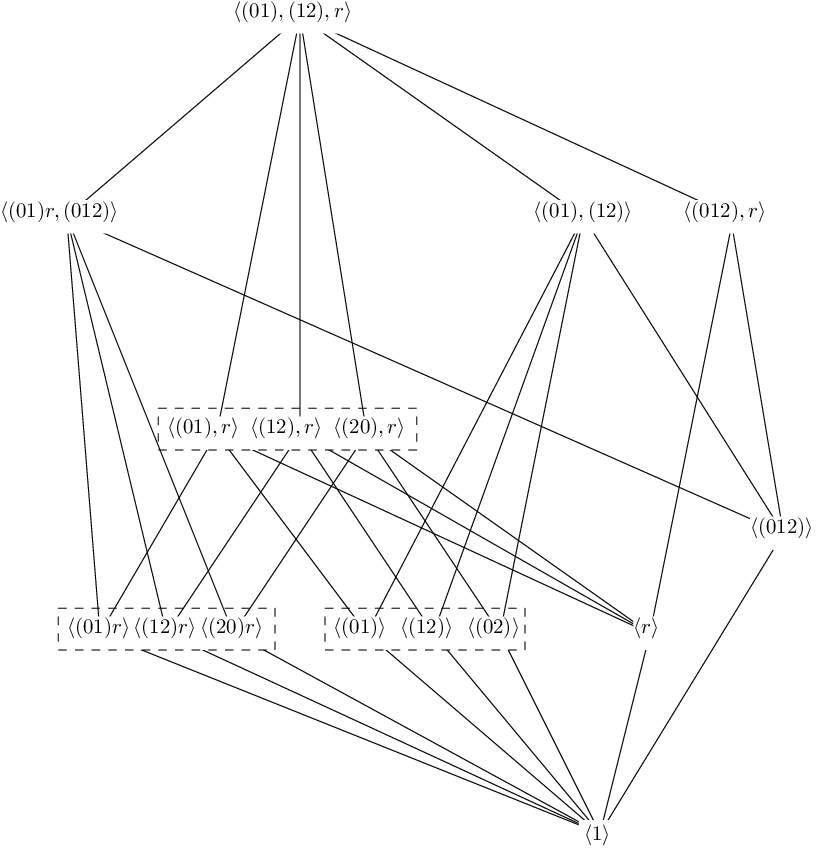}
\caption{The lattice of $S_3R$.}
\label{fig:S3R_lattice}
\end{figure}

\subsection{The Orbits of the Subgroups of $S_3R$ acting on $\Sigma^n$}
\label{sec:orbis_S3R}

The following lemma facilitates the calculation of orbits of
subgroups that contain the permutation $(012)$.

\begin{prop}
\label{perm_012}
Suppose $H$ is a subgroup of $S_3R$ that contains the permutation $(012)$, and suppose
that $A$ is an orbit of $H$ acting on $\Sigma^n$.
Then the number 3 divides the length of $A$.
\end{prop}

\begin{proof}
There is no word in $\Sigma^n$ that is invariant under $(012)$.
This is equivalent to saying that for all words $w$ in $\Sigma^n$,
the permutation $(012)$ is not an element of the subgroup $\stb(w)$, and hence
that the group $\langle (012) \rangle$ is not a subgroup of $\stb(w)$.
Since $(012)$ and $(210)$ are the only elements of order 3 in $H$, the
subgroup $\stb(w)$ is not divisible by 3.
From $|\stb(w)| \ |A| = |H|$ follows the proposition.
\end{proof}

The next simple lemma will help us in classifying orbits of degree 1.

\begin{prop}
\label{S_3R_non_free}
Suppose $H$ is a subgroup of $S_3R$ and $A$ is an orbit of $H$ acting on $\Sigma^n$.
Then $\deg(A) \le 1$
\begin{enumerate}[(i)]
        \item if $(01) \in H$ and there is a $w \in A$ such that $w=(01)w$; or
        \item if $(01)r \in H$ and there is a $w \in A$ such that $w = (01)rw$.
\end{enumerate}
\end{prop}

\begin{proof}
Let $\alpha$ be $(01)$ or $(01)r$ and suppose that $f$ is an invariant function
from $A$ to $\Sigma$.
Then the relation $f(w) = \alpha^{-1} f(\alpha w) = \alpha f(w)$ implies $f(w)=2$.
\end{proof}

All words below are understood to be words over $\Sigma$.
The word $w$ always denotes a word of $\Sigma^n$.
Sometimes we will write $w$ as a concatenation of two words, that is $w = uv$, if $n=2m$ is even,
and as a concatenation of a word, a symbol, and a further word, that is $w = uav$, if $n=2m+1$ is odd.
If we do so, we assume that $|u| = |v| = m$.
For each subgroup $H$ of $S_3R$ the number of free orbits of $H$ acting on $\Sigma^n$
is calculated as follows.
In the calculations themselves we will make frequent use of
Lemma~\ref{freeorb} and Lemma~\ref{S_3R_non_free}, without explicitly referencing them.

\begin{enumerate}

\item The group $\langle 1 \rangle$.
$\Sigma^n$ splits up into $3^n$ orbits of length 1.

\item The group $\langle r \rangle$.
A calculation similar to the state set of size 2 yields
$(3^{2m} + 3^m)/2$ orbits of degree 3 if $n=2m$ is even, and
$(3^{2m+1} + 3^{m+1})/2$ orbits of degree 3 if $n=2m+1$ is odd.

\item The groups $\langle (01) \rangle$, $\langle (12) \rangle$, $\langle (20) \rangle$.
We study  $\langle (01) \rangle$.
Only the word $w=2^n$ ($n$ copies of $2$) satisfies the equation $(01)w = w$.
Hence the orbit $\{2^n\}$ is of degree 1.
The remaining words split up into $(3^n-1)/2$ orbits of length 2 and degree 3.

\item The groups $\langle (01)r \rangle$, $\langle (12)r \rangle$, $\langle (20)r \rangle$.
We study  $\langle (01)r \rangle$.

Suppose $n = 2m$ is even and $w=uv$.
The relation $uv = (01)r(uv) = ((01)rv)((01)ru)$ implies $v=(01)ru$ and so is
satisfied by $3^m$ words which form $3^m$ orbits of length 1 and degree 1.
The remaining words split up into $(3^{2m} - 3^m)/2$ orbits of degree 3.

Suppose $n=2m+1$ is odd and $w=uav$.
The relation $uav = (01)r(uav) = ((01)rv)((01)a)((01)ru)$ implies $v=(01)ru$ and $a=2$, which is
satisfied by $3^m$ words resulting in  $3^m$ orbits of length 1 and degree 1.
The remaining words split up into $(3^{2m+1} - 3^m)/2$ orbits of degree 3.

\item The group $\langle (012) \rangle$.
All orbits are of the form $\{w, (012)w, (210)w \}$ with pairwise different elements.
This shows that $\Sigma^n$ partitions into $3^{n-1}$ orbits of degree 3.

\item The group $S_3 = \langle (01), (12) \rangle$.
The orbit $\{0^n, 1^n, 2^n\}$ is of degree 1 because $(01)2^n = 2^n$.
The remaining words split up into $(3^{n-1}-1)/2$ orbits of length 6 and degree 3.

\item The groups $\langle (01), r \rangle$, $\langle (12), r \rangle$, $\langle (20), r \rangle$.
We study $\langle (01), r \rangle$.
Only the word $2^n$ satisfies the relation $w = rw = (01)w = (01)rw$.
The corresponding orbit $\{2^n\}$ is of degree 1.

Suppose $n = 2m$ is even.
There are two ways that the orbit $\{w$, $rw$, $(01)w$, $(01)rw \}$
can fold up into orbits of length 2.

First, if $w=rw$ and $w\neq (01)w$ holds.
From the set of $3^m$ words that satisfy $w=rw$ we remove $2^n$.
The remaining words in this set split up into $(3^m-1)/2$
orbits of length 2 and degree 3.

Second, if $rw = (01)w$ and $rw \neq w$ holds.
If $w=uv$ the relation $rw = (01)w$ implies $v=(01)ru$.
As above, we remove from the set of $3^m$ words that satisfy this relation the word $2^n$
to obtain $(3^m-1)/2$ orbits of length 2 and degree 1.

The remaining $3^{2m} - 2(3^m-1) - 1$ words in $\Sigma^n$ split up
into orbits of length 4 and degree 3.
Summing up the orbits of degree 3, we obtain for their number
$(3^{2m} -2(3^m - 1) -1)/4 + (3^m-1)/2 = (3^{2m}-1)/4$.

Suppose $n=2m+1$ is odd.
We consider again the two different types of orbits of length 2.

The first occurs, if $w=rw$ and $w\neq (01)w$ holds.
A similar calculation as above obtains
$(3^{m+1}-1)/2$ orbits of length 2 and degree 3.

The second occurs, if $rw = (01)w$ and $rw \neq w$.
If $w$ is written as $uav$, the relation
becomes $(rv)a(ru) = ((01)u)((01)a)((01)v)$, yielding the constraints $v=(01)ru$ and $a=2$
which are satisfied by $3^m$ words.
Removing the word $2^n$ from this set results in $(3^m-1)/2$ orbits of length 2 and degree 1.

The remaining  $3^{2m+1} - (3^{m+1}-1) - (3^m-1) - 1 = (3^{m+1}-1)(3^m-1)$ words
in $\Sigma^n$ split up
into free orbits of length 4.
Hence the total number of orbits of degree 3 is
$(3^{m+1}-1)(3^m-1)/4 + (3^{m+1}-1)/2 = (3^{m+1}-1)(3^m+1)/4$.

\item The group $\langle (012)r \rangle$.
If $w = rw$, then $\{ w, (012)w, (210)w \}$ is an orbit of length 3.
If $w \neq rw$, the orbit containing $w$ is of length 6.
The calculation is similar to the one of the group $\langle r \rangle$.
If $n=2m$ is even, $\Sigma^n$ partitions into $(3^{2m-1} + 3^{m-1})/2$ orbits of degree 3,
if $n=2m+1$ is odd, $\Sigma^n$ partitions into $(3^{2m} + 3^{m})/2$ orbits of degree 3.

\item The group $\langle (01)r, (012) \rangle$.
An orbit is of length 3 if and only if the orbit contains a word $w$
that satisfies the relation $w = (01)rw$.

Suppose $n=2m$ is even and write $w=uv$.
Then $w=(01)rw$ holds if and only if $v=(01)ru$ holds.
Each of these $3^m$ words is in a different orbit.
This shows that there are $3^m$ orbits of length 3.
All of them have degree 1.
The remaining words split up into $(3^{2m-1}-3^m)/2$ orbits of length 6 and degree 3.

Suppose $n=2m+1$ is odd and write $w=uav$.
Then $w=(01)rw$ holds if and only if $v=(01)ru$ holds and $a=2$.
Again, this shows that there are $3^m$ orbits of length 3 and degree 1.
The remaining words split up into $(3^{2m}-3^m)/2$ orbits of length 6 and degree 3.

\item The group $S_3R$.
\label{orbits_S3R}
The set $\{0^n,1^n,2^n\}$ is the only orbit of length 3. The degree is 1.
If $n=1$ it is the only orbit.

Suppose $n=2m$ is even.
By Lemma~\ref{perm_012} the next possible length of an orbit is 6.
If $w$ is a word in an orbit of length 6, either $w = rw$ or
$w = \alpha w$, where $\alpha$ denotes a transposition.
There are $3^m$ words for which $w=rw$.
If we subtract the 3 words of the orbit of length 3, we obtain
$3(3^{m-1}-1)$ words that split up into $(3^{m-1}-1)/2$ orbits of degree 3.
There are also $3^m$ words for which $w=(01)rw$, or
$3 \times 3^m$ words for which $w= \alpha rw$, where $\alpha$ denotes any of
the three transpositions $(01), (12), (20)$.
If we subtract the 3 words of the orbit of length 3, we obtain
$3(3^{m}-1)$ words that split up into $(3^{m}-1)/2$ orbits of degree 1.

The remaining words of $\Sigma^n$ partition into orbits of length 12.
The number of remaining words is
$3^{2m} - 3(3^{m-1}-1) - 3(3^{m}-1) - 3$, splitting up into
$(3^m-1)(3^{m-1}-1)/4$ orbits of length 12 and degree 3.
The total number of orbits of degree 3 then is $(3^m-1)(3^{m-1}-1)/4 + (3^{m-1}-1)/2
= (3^m+1)(3^{m-1}-1)/4$.

Suppose $n > 1 $ and $n=2m+1$ is odd.
There are $3^{m+1}$ words that satisfy $w=rw$.
Similar to the calculation above,
we obtain $3^{m+1} -3$ words that split up into $(3^m-1)/2$ orbits of degree 3.
The relation $ucv = (01)r(ucv)$ implies $v=(01)r$ and $c=2$.
As above,
we see that $3^{m+1} - 3$ words satisfying $w = \alpha r w$ split up into $(3^{m} -1)/2$
orbits of degree 1.

The remaining words of $\Sigma^n$ partition into orbits of length 12.
The number of remaining words is
$3^{2m+1} - 6(3^m-1) - 3$, splitting up into
$(3^m-1)^2/4$ orbits of length 12 and degree 3.
The total number of orbits of degree 3 then is $(3^m-1)^2/4 + (3^m-1)/2 = (3^{2m}-1)/4$.
Fig.~\ref{fig:orbits_s3r} depicts the orbits of $\Sigma^3 / S_3R$.
The graphs are informal, highlighting certain symmetries of the orbits.

\begin{figure}
\centering
\includegraphics[scale=0.7]{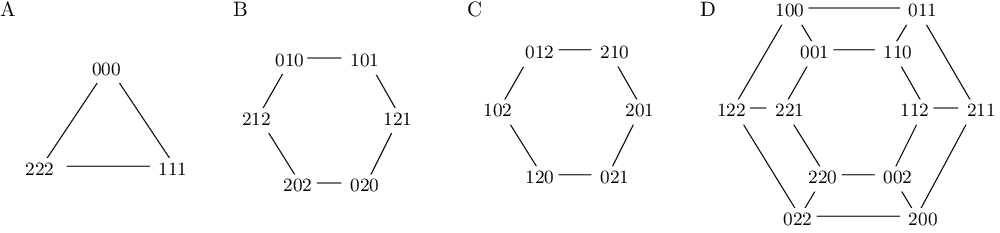}
\caption{The orbits of $S_3R$ acting on $\{0,1,2\}^3$.}
\label{fig:orbits_s3r}
\end{figure}

\end{enumerate}

\subsection{Counting the Orbits of $L(3,n) / S_3R$ by Type}

Too shorten the notation, we enumerate the representatives of the conjugation classes:
$H_1 = S_3R$, $H_2=\langle (01)r, (012) \rangle$, $H_3=\langle (01), (12) \rangle$,
$H_4 = \langle (012)r \rangle$, $H_5 = \langle (01), r \rangle$,
$H_6 = \langle (012) \rangle$, $H_7=\langle (01)r \rangle$,
$H_8=\langle (01) \rangle$, $H_9 = \langle r \rangle$, and $H_{10} = \langle 1 \rangle$.
We set $a_i =|\fix(H_i)|$, $b_i = |\stb^{-1}(H_i)|$, and $c_i = |\orb(3,n,H_i)|$.
In the previous subsection we determined
the degree of orbits for one representative $H_i$ of the conjugation classes $[H_i]$.
If $p_i$ is the number of orbits of $H_i$ of degree 3, then $a_i = 3^{p_i}$.
We express now $b_i$ and $c_i$ in terms of $a_i$. Then
\begin{eqnarray*}
& & b_1 = a_1 = c_1; \\
& & b_2 = a_2 - a_1, c_2 = b_2 / 2; \\
& & b_3 = a_3 - a_1, c_3 = b_3 / 2; \\
& & b_4 = a_4 - a_1, c_4 = b_4 / 2; \\
& & b_5 = a_5 - a_1, c_5 = 3 b_5 / 3 = b_5; \\
& & b_6 = a_6 - a_1 - b_2 - b_3 - b_4 = a_6 + 2 a_1 - a_2 - a_3 - a_4, c_6 = b_6 / 4; \\
& & b_7 = a_7 - a_1 - b_2 - b_5 = a_7 + a_1 - a_2 - a_5, c_7 = 3 b_7 / 6 = b_7 / 2;\\
& & b_8 = a_8 - a_1 - b_3 - b_5 = a_8 + a_1 - a_3 - a_5, c_8 = 3 b_8 / 6 = b_8 / 2;\\
& & b_9 = a_9 - a_1 - b_4 - 3 b_5 = a_9 + 3a_1 -a_4 - 3a_5, c_9 = b_9 / 6;\\
& & b_{10} = a_{10} - a_1 - b_2 - b_3 - b_4 - 3 b_5 - b_6 - 3 b_7 - 3 b_8 - b_9 \\
& & \ \ = a_{10} -6a_1 + 3a_2 + 3a_3 + a_4 +6a_5 - a_6 - 3a_7 -3a_8 - a_9, c_{10} = b_{10} / 12.
\end{eqnarray*}

The equations for $c_5$, $c_7$, and $c_8$ contain an additional factor 3 due to
the size of the conjugation classes, see Equation~\ref{eq:orb_H_card} and Fig.~\ref{fig:S3R_lattice}.
The total number of orbits is therefore
\begin{eqnarray}
\label{eq:s3_tot_sum_orb}
\sum_{i=1}^{10} c_i = \frac{1}{12} \left(a_9 +a_{10} \right) + \frac{1}{6}\left(a_4+a_6\right)
+ \frac{1}{4} \left(a_7+ a_8\right).
\end{eqnarray}

\begin{prop2}
\label{prop:s3r}
Let $m$ be a nonnegative integer.
The set $L(3,2m+1)$ of local rules partitions into
\begin{eqnarray*}
& & \frac{1}{12} \left(3^{(3^{2m+1} + 3^{m+1})/2} + 3^{3^{2m+1}}\right)
+ \frac{1}{6}\left( 3^{(3^{2m} + 3^{m})/2} +3^{3^{2m}} \right)
 +  \frac{1}{4} \left(3^{(3^{2m+1} - 3^m)/2}+ 3^{(3^{2m+1}-1)/2}\right)
\end{eqnarray*}
orbits.

Let $m$ be a positive integer.
The set $L(3,2m)$ of local rules partitions into
\begin{eqnarray*}
&  & \frac{1}{12} \left(3^{(3^{2m} + 3^m)/2} + 3^{3^{2m}}\right)
+\frac{1}{6}\left( 3^{(3^{2m-1} + 3^{m-1})/2}  + 3^{3^{2m-1}} \right)
+ \frac{1}{4} \left(3^{(3^{2m} - 3^m)/2}+ 3^{(3^{2m}-1)/2}\right)
\end{eqnarray*}
orbits.
\end{prop2}

The total number of orbits can be derived more simply from Burnside's Lemma:
\begin{eqnarray*}
& & |L(3,n) / S_3R| = \frac{1}{|S_3R|} \sum_{\alpha \in S_3R} |\fix(\alpha)| = \\
& & \, \, \, \, \, \, \frac{1}{12} \left(
|\fix(1)| + |\fix(r)| + 3 |\fix((01))| + 3 |\fix((01)r)| + 2|\fix((012))| + 2|\fix((012)r)|
\right),
\end{eqnarray*}
where we have used the relation $|\fix(\alpha)| = |\fix( \beta \alpha \beta^{-1})|$.
If we note that $\fix(\alpha) = \fix( \langle \alpha \rangle)$,
we obtain Equation~(\ref{eq:s3_tot_sum_orb}).
Since the calculation depends only on the cyclic subgroups, Burnside's lemma is preferable,
if only the total number of orbits is required.

\begin{table}
\begin{center}
\caption{Count of three-state orbits by type.}
\label{tab:s3n2n3}
\begin{tabular}{|l|r|r|r|}
\hline
$H$ & $|\orb(3,1,H)|$ & $|\orb(3,2,H)|$ & $|\orb(3,3,H)|$ \\ \hline
$\langle (01), (12), r \rangle$ & 1 &       1 &                   9  \\ \hline
$\langle (01)r, (012) \rangle$ 	& 0 &       0 &                   9  \\ \hline
$\langle (01), (12) \rangle$ 	& 0 &       1 &                  36  \\ \hline
$\langle (012)r \rangle$ 		& 1 &       4 &                 360  \\ \hline
$\langle (01), r \rangle$ 		& 2 &       8 &              6\,552  \\ \hline
$\langle (012) \rangle$ 		& 0 &       4 &             4\,716   \\ \hline
$\langle (01)r \rangle$ 		& 0 &       9 &           262\,431   \\ \hline
$\langle (01) \rangle$ 			& 0 &      35 &           793\,845   \\ \hline
$\langle r \rangle$ 			& 3 &     116 &       64\,566\,684   \\ \hline
$\langle 1 \rangle$ 			& 0 &  1\,556 & 635\,433\,642\,324   \\ \hline
								& 7 &  1\,734 & 635\,499\,276\,966   \\ \hline
\end{tabular}
\end{center}
\end{table}

Table \ref{tab:s3n2n3} lists the cardinalities of orbits by type,
i.e. $|\orb(3,n, H)|$, $[H] \in {\cal C}(S_3R)$,
for a neighbourhood size of one, two, and three.
The last row gives the total number of orbits, that is $|L(3, n) / S_3R|$.
In contrast to two-state CAs,
we have refrained from giving explicit formulas for $c_i=|\orb(3,n,H_i)|$ in the above proposition.
These formulas become lengthy, but can be easily derived by expressing the $c_i$’s
in terms of the $p_i$'s above.

\subsection{Constructing local rules that are invariant}
The focus of this study so far has been on deriving formulas for the cardinalities of the orbits of
CA rules by type as well as for their total number.
We point out that an analogous method to the one described
in Subsection~\ref{subsec:outline_calc} can be used to actually construct the local rules that are invariant.
We will not treat this subject systematically, but give an example.
Let $\Sigma = \{0,1,2\}$ and $n=3$.
The orbits of $S_3R$ acting on $\Sigma^3$ and their degrees were derived in Subsection~\ref{sec:orbis_S3R} and
depicted in Fig.~\ref{fig:orbits_s3r}.
We have also seen that the degree of orbit $A$ and $C$ is 1 and the degree of orbit $B$ and $D$ is 3.
According to Subsection~\ref{sec:sym_op_domain}, we form the set of local rules $\{f\}$ on $\Sigma^3$
by considering the restriction of $f$ on the partitions $A, \ldots, D$, denoted by $f_A, \ldots, f_D$.
Since the degree of $A$ is 1, we need to determine the value of $f_A$ for a given word $w$,
for instance $000$.
The relation $(12)f_A(000)=f_A((12)000) = f_A(000)$ implies $f_A(000) = 0$.
The other values follow from $f_A(\alpha 000) = \alpha f_A(000) = \alpha 0$, $\alpha \in S_3R$.
Similarly, $(02)rf_C(012) = f_C((02)r 012) = f_C(012)$ implies $f_C(012) = 1$, the other values
follow as above.
Thus there is only one function $f_A$ on $A$ and only one function $f_C$ on $C$ invariant under $S_3R$:
\small
\[
f_A =
\begin{pmatrix}
000 & 111 & 222 \\
0 & 1 & 2
\end{pmatrix},  \ \
f_C =
\begin{pmatrix}
012 & 210 & 201 & 102 & 120 & 021 \\
 1  &  1  &  0  &  0  &  2  &  2
\end{pmatrix}.
\]
\normalsize
The orbit $B$ is of degree 3, so we can pick a word in $B$, say $010$, define $f_B(010)=a$ with
$a \in \Sigma$ arbitrary, and derive the other function values as above.
We proceed with orbit $D$ in the same way and set $f_D(001) = b$, $b \in \Sigma$, so
\small
\[
f_B =
\begin{pmatrix}
010 & 101   & 121    & 212   & 020    & 202 \\
 a  & (01)a & (012)a & (02)a & (12)a  & (021)a
\end{pmatrix}, \ \ \mbox{and}
\]
\normalsize
\setcounter{MaxMatrixCols}{12}
\small
\[
f_D =
\begin{pmatrix}
001 & 100 & 011   & 110   & 112    & 211    & 002   & 200   & 022    & 220    & 122   & 221   \\
 b  &  b  & (01)b & (01)b & (012)b & (012)b & (12)b & (12)b & (021)b & (021)b & (02)b & (02)b
\end{pmatrix}.
\]
\normalsize
We set $f(w) = f_X(w)$ if $w \in X$.
Since $a$ and $b$ were arbitrary symbols of $\Sigma$, there are 9 orbits of length 1 of type
$S_3R$, in accordance with the entry $|\orb(3,3,S_3R)|$ in Table~\ref{tab:s3n2n3}.

\section{Validation}
\label{validation}

In the previous sections the exact numbers of orbits for one-dimensional
two-state and three-state CAs were calculated.
This section takes another approach and describes an algorithmic brute-force approach to determine
these numbers for small $k$ and $n$.
The algorithm is implemented in Python 3 and depicted in Table~\ref{py_pgm}.

\begin{table}
\caption{Python 3 program that calculates the number of orbits.}
\label{py_pgm}
\begin{scriptsize}
\begin{verbatim}
import itertools as it
k = 3; n = 2                    # number of states; neighbourhood size
s = tuple(range(k))             # state set (0,1,..,k-1)
def enc(w) :                    # encodes a word
    v = 0
    for a in w : v = k*v+a
    return v                    # returns  w[0]*k^(n-1)+..+w[n-1]
rfl_pairs = [(enc(w),enc(w[::-1])) for w in it.product(s, repeat=n)
    if w != w[::-1]]            # pairs (i,j), w_j = rw_i != w_i
def reflectRule(f) :            # returns reflected rule
    g = list(f)                 # copy f
    for (i,j) in rfl_pairs : g[i] = f[j]
    return tuple(g)
def permutateRule(f, perm) :    # returns permutated rule
    g = [0] * k**n
    for w in it.product(s, repeat=n) :
        g[enc([perm[a] for a in w])] = perm[f[enc(w)]]
    return tuple(g)
def orbit(f) :                  # returns the orbit of f
    orb = set()
    for perm in it.permutations(s) :
        pf = permutateRule(f,perm)
        orb.update({pf, reflectRule(pf)})
    return tuple(orb)
def countOrbits() :             # counts all orbits
    processed = set()           # keep track of processed rules
    count = 0
    for f in it.product(s,repeat=k**n) :
        if f not in processed :
            count += 1
            processed.update(orbit(f))
    return count
print(countOrbits())            # prints number of orbits
\end{verbatim}
\end{scriptsize}
\end{table}

We start with some general considerations applicable to any programming language that
supports arrays (referred to as sequences in Python).
The state set $\Sigma = \{0,1,\ldots,k-1\}$ is ordered, and so is
the set $\Sigma^n$ if we adopt the lexicographical order.
We write the set $\Sigma^n$ as an increasing sequence $\left( w_i \right)$; $0 \leq i < k^n$.
This arrangement allows for the representation of a local rule $f$ by the
sequence $(b_i)$, $0 \leq i < k^n$, where $b_i = f(w_i)$.
We define an encoding function, denoted by $\enc$, which maps a word
to an integer.
The function returns the index $i$ of the word $w$ within the sequence $(w_i)$ such that $w = w_i$,
or equivalently, the
numerical value when $w$ is read as a number in base $k$:
if $w = a_{n-1} \ldots a_0$, then $\enc(w) = a_{n-1} k^{n-1} + \cdots + a_0$.
Given a local rule $f$ represented by the sequence $b = (b_i)$ and a word $w$,
to find $f(w)$ compute $j = \enc(w)$, and then access the $j$-th element in
the sequence $b$: $f(w) = b_j$.

The program listed in Table~\ref{py_pgm} implements local rules and words as sequences.
The symmetry operators are implemented in a manner closely aligned with their theoretical definitions.

We first discuss  the reflection operator.
The variable \texttt{rfl\_pairs} refers to a sequence of integer pairs $(i,j)$ satisfying
the relations $w_j = rw_i$ and $w_i \ne rw_i$
 (\texttt{w[::-1]} is a Python idiom used to reverse a sequence).
The function \texttt{reflectRule} takes a sequence \texttt{f} that represents a rule, and returns
its reflected version \texttt{g}.
Initially, the rule provided as argument is copied into the variable \texttt{g}.
Then, a \texttt{for} loop iterates over \texttt{refl\_pairs}, modifying
the values of \texttt{g} accordingly to the pairs of \texttt{reflectRule}.

We now shift focus to the implementation of the permutation operator.
Permutations of the state set are represented as sequences of length \texttt{k}.
The function \texttt{permutateRule} accepts a local rule \texttt{f} and a permutation
\texttt{perm}, and returns
the permutated rule \texttt{g}.
It begins with initializing the variable \texttt{g} with a sequence of length $k^n$.
A \texttt{for} loop then iterates all words in the domain, and
for each word \texttt{w}, \texttt{g} is changed, according to the
equation $g(\sigma(a_0\ldots a_{n-1})) = g\left( (\sigma a_0)\ldots (\sigma a_{n-1}) \right) = \sigma f(a_0\ldots a_n)$.

The function \texttt{orbit} determines the orbit of the input function \texttt{f}.
It iterates all permutations of $\Sigma$, and adds the permutated rule and
the permutated and reflected rule to the orbit.
If $f$ is invariant under a certain operation then $f$ will not be changed.
Since the underlying data structure of the equivalence class is a \texttt{set},
subsequent additions of the same element have no effect.

Lastly \texttt{countOrbits} iterates the set of local rules:
\texttt{it.product(s,repeat=k**n)} creates the cartesian product of $\Sigma$ with itself, $k^n$ times,
representing the set of local rules.
If a rule belongs to an orbit of an already processed rule, the body of the
\texttt{for} loop is skipped.
Otherwise a new orbit is created, the counter is incremented, and the members of the class are
stored in a set referenced by the variable \texttt{processed}.

On a typical PC, the program prints the result within a few seconds for the
input parameters $k=2$ and $n\leq 4$ as well as for
$k=3$ and $n \leq 2$.
With an optimized implementation and improved hardware,
it might be possible to achieve results for a few additional combinations,
such as $k=2$, and $n=5$.
However, the algorithm's runtime complexity prevents
calculations for larger input parameters.

The presented implementation is minimalistic.
We briefly explore two kinds of improvements.
\begin{enumerate}
\item
By incorporating minor changes, more detailed insights about the orbits can be
obtained.
As an example, adding a hashtable to the function body of \texttt{countEquiClasses}
allows tracking the number of equivalence classes based on their size.

\item Although traversing the entire domain might be unfeasible,
exploring only parts of it might still be instructive.
For instance, if the main loop is adjusted to iterate only the set $\fix(\langle (012) \rangle)$,
it is feasible to obtain the results for the upper lattice of $S_3R$,
which consists of groups that encompass the group
$\langle (012) \rangle$.

\end{enumerate}

\section{Summary}
\label{sec:summary}

This work investigates the classification of one-dimensional cellular automata (CAs) into
orbits (also called equivalence classes) using a group-theoretical approach.
A cellular automaton operates on a bi-infinite lattice of cells, each existing in one of a finite number of states, and evolves according to local rules that depend on a fixed-size neighbourhood of cells.

The study defines equivalence through transformations such as reflection, permutation of states, and their combinations.
The key contributions include:

\begin{itemize}
    \item Formalizing orbits by systematically incorporating symmetry operations, including reflection
and state permutations, to identify and group equivalent rules within the set of local rules;

    \item Deriving orbits for two-state and three-state cellular automata with arbitrary neighbourhood,
which generalizes previous results and corroborates existing findings, such as the well-established 88 equivalence
classes for elementary cellular automata (two states, three neighbours);

    \item Exploring group actions and symmetries acting on the set of local rules,
developing a comprehensive methodological framework for calculating orbits of
the set of local rules across varying numbers
of states and neighbourhood sizes;

	\item Classifying orbits by their isomorphism type with respect to the symmetry operations
and giving results for the number of orbits per type;

    \item Implementing an algorithmic validation through a brute-force approach in Python,
empirically verifying the theoretical results for families of CAs with a small
set of local rules.
\end{itemize}

The study concludes by highlighting the significance of symmetry-based classification in substantially reducing the number of unique CA rules.
This approach provides a rigorous foundation for future investigations into cellular automata dynamics and computational properties,
potentially opening new avenues for understanding discrete complex systems and computational mechanisms.

\subsection*{Acknowledgments}

The authors would like to express their sincere thanks to the anonymous referees for their valuable comments and suggestions.

This research was funded in whole or in part by the Austrian Science Fund (FWF) [Grant DOI:10.55776/PIN5424624].
The authors acknowledge TU Wien Bibliothek for financial support through its Open Access Funding Programme.


%

\end{document}